\documentclass[aps,prl,10pt,twocolumn,superscriptaddress,floatfix,nofootinbib,showpacs,longbibliography]{revtex4-1}
\usepackage{makecell}
\usepackage{array}
\newcolumntype{C}[1]{>{\centering\arraybackslash}m{#1}}
\usepackage{booktabs,adjustbox}
\usepackage{float}
\usepackage{amsmath}
\usepackage[utf8]{inputenc}  
\usepackage[T1]{fontenc}     
\usepackage[table]{xcolor}
\usepackage{color,soul}
\usepackage[british]{babel}  
\usepackage[colorlinks=true, citecolor=blue, urlcolor=blue]{hyperref}  
\usepackage{graphicx} 
\usepackage[babel]{microtype}  
\usepackage{amsmath,amssymb,amsthm,bm,amsfonts,mathrsfs,bbm} 

\usepackage{xspace}  
\usepackage{xcolor}
\usepackage{multirow}
\usepackage{bigstrut}
\usepackage{braket}
\usepackage{color}
\usepackage{natbib}
\usepackage{multirow}
\usepackage{mathtools}
\usepackage{float}
\usepackage{blkarray}
\usepackage[caption = false]{subfig}
\usepackage{xcolor,colortbl}
\usepackage{color}
\usepackage{rotating}
\usepackage{makecell}

\usepackage{tabularx}

\definecolor{revised}{rgb}{.9,0,.3}

\newcommand{\be}{\begin{equation}}
\newcommand{\ee}{\end{equation}}
\newcommand{\ba}{\begin{eqnarray}}
\newcommand{\ea}{\end{eqnarray}}

\def\>{\rangle}
\def\<{\langle}

\usepackage{centernot}
\usepackage{subfig}

\begin{document}

\author{Chen Ding}
\thanks{These three authors contributed equally.}
\affiliation{Henan Key Laboratory of Quantum Information and Cryptography, Zhengzhou, Henan 450000, China}
\author{Edwin Peter Lobo}
\thanks{These three authors contributed equally.}
\affiliation{Laboratoire d'Information Quantique, Université libre de Bruxelles (ULB), Av. F. D. Roosevelt 50, 1050 Bruxelles, Belgium.}

\author{Mir Alimuddin}
\thanks{These three authors contributed equally.}
\affiliation{Department of Physics of Complex Systems, S. N. Bose National Center for Basic Sciences, Block JD, Sector III, Salt Lake, Kolkata 700106, India.}
\affiliation{ICFO-Institut de Ciencies Fotoniques, The Barcelona Institute of Science and Technology, Av. Carl Friedrich Gauss 3, 08860 Castelldefels (Barcelona), Spain.}
\author{Xiao-Yue Xu}
\affiliation{Henan Key Laboratory of Quantum Information and Cryptography, Zhengzhou, Henan 450000, China}
\author{Shuo Zhang}
\affiliation{Henan Key Laboratory of Quantum Information and Cryptography, Zhengzhou, Henan 450000, China}
\author{Manik Banik}
\email{manik11ju@gmail.com}
\affiliation{Department of Physics of Complex Systems, S. N. Bose National Center for Basic Sciences, Block JD, Sector III, Salt Lake, Kolkata 700106, India.}

\author{Wan-Su Bao}
\email{bws@qiclab.cn}
\affiliation{Henan Key Laboratory of Quantum Information and Cryptography, Zhengzhou, Henan 450000, China}
\author{He-Liang Huang}
\email{quanhhl@ustc.edu.cn}
\affiliation{Henan Key Laboratory of Quantum Information and Cryptography, Zhengzhou, Henan 450000, China}

\title{Quantum Advantage: A Single Qubit's Experimental Edge in Classical Data Storage}


\begin{abstract}
We implement an experiment on a photonic quantum processor establishing efficacy of the elementary quantum system in classical information storage. The advantage is established by considering a class of simple bipartite games played with the communication resource qubit and classical bit (c-bit), respectively. Conventional wisdom, supported by the no-go theorems of Holevo and Frenkel-Weiner, suggests that such a quantum advantage is unattainable when the sender and receiver share randomness or classical correlations. However, our results reveal a quantum advantage in a scenario devoid of any shared randomness. Our experiment involves the development of a variational triangular polarimeter, enabling the realization of positive operator value measurements crucial for establishing the targeted quantum advantage. Beyond showcasing a robust communication advantage with a single qubit, our work paves the way for immediate applications in near-term quantum technologies. It provides a semi-device-independent certification scheme for quantum encoding-decoding systems and offers an efficient method for information loading and transmission in quantum networks.
\end{abstract}

\maketitle
{\it Introduction.--} The second quantum revolution harnesses the non-classical properties of elementary quantum systems to develop novel technologies that surpass their classical counterparts \cite{Dowling2003, Deutsch2020, HuangSuperconducting2020, Aspect2023, HuangNear2023}. Recent advancements have enabled the successful implementation of various protocols using quantum devices, leading to significant progress in communication \cite{Gisin2007, Ren2017, Yin2017}, metrology \cite{Giovannetti2011, Pezze2018, Braun2018}, and computational tasks \cite{Kok2007, Ladd2010, Zhong2020}. As quantum devices become increasingly sophisticated and precise, they facilitate the exploration of cutting-edge protocols where quantum resources offer distinct advantages.

In this work, we present a state-of-the-art protocol demonstrating the communication advantage of a single elementary quantum system over a classical $1$-bit channel, in the absence of pre-shared randomness between sender and receiver, as also noted in \cite{Patra2024}. We achieved this by constructing a single-photon source and developing a variational triangular polarimeter, which enables us to flexibly perform positive operator value measurement (POVM) of various bases on the photons. We conduct a series of experiments, which yield a squared statistical overlap of \(0.9998 \pm 0.0016\) 
averaged over the games, being excellent agreement with theoretical predictions. The quantum protocols significantly outperformed the best classical strategies, as evidenced by their superior performance in a suitably chosen `satisfaction function'. Our results confirm a robust quantum communication advantage using a single qubit, with immediate implications for near-term quantum technologies. Notably, the quantum advantage is achieved through non-orthogonal encodings at the sender's end and non-projective measurements at the receiver's end, contributing to a semi-device-independent non-classicality certification scheme for quantum encoding-decoding systems.

{\it The Set-Up.--} The simplest communication scenario involves two distant parties, Alice and Bob, each equipped with finite alphabets \( X = \{x\} \) and \( Y = \{y\} \), containing \( m \) and \( n \) letters, respectively. Here, we examine the set of conditional probability distributions \( \{p(y|x)\} \) generated by transmitting an elementary system, such as a qubit, a classical bit, or a system described by generalized probability theory (GPT) \cite{Barrett07, Janotta11, Janotta14, Brunner14, Banik19, Bhattacharya20}, from Alice to Bob; $p(y|x)$ represents the probability of outcome \( y \) given input \( x \). Alice encodes \( x \) onto an ensemble of states and transfers the system to Bob, who then determines the output letter \( y \) by performing a measurement. Alice and Bob may use probabilistic encoding and decoding strategies by locally randomizing their respective deterministic protocols. The resulting correlations \( \{p(y|x)\} \) can be represented as points in \( \mathbb{R}^{m \times n} \). In this work, we focus on the sets of points \( \mathcal{C}^{m \times n}_d \) and \( \mathcal{Q}^{m \times n}_d \), which are obtained when Alice sends a \( d \)-dimensional classical or quantum system, respectively, to Bob. In a recent study by Frenkel and Weiner \cite{Frenkel2015} (see also \cite{Dallarno2017,Naik2022,Sen2022,Patra2023}), it was demonstrated that
$\text{CvHul}(\mathcal{C}^{m \times n}_d) = \text{CvHul}(\mathcal{Q}^{m \times n}_d)~\forall~ d, m, n$; where \(\text{CvHul}(S)\) of a set $S$ is the minimal convex set that contains $S$. In other words, assuming the availability of pre-shared classical correlations, or shared randomness, as a free resource, Frenkel-Weiner's result establishes that a quantum channel can be replaced by a classical channel of the same dimension. 
However, in certain scenarios, additional assumptions about shared randomness may be necessary, and treating it as a free resource might not be justified. This issue is illustrated through a communication game discussed below.

\begin{figure}[t]
\begin{center}
\includegraphics[width=0.95\linewidth]{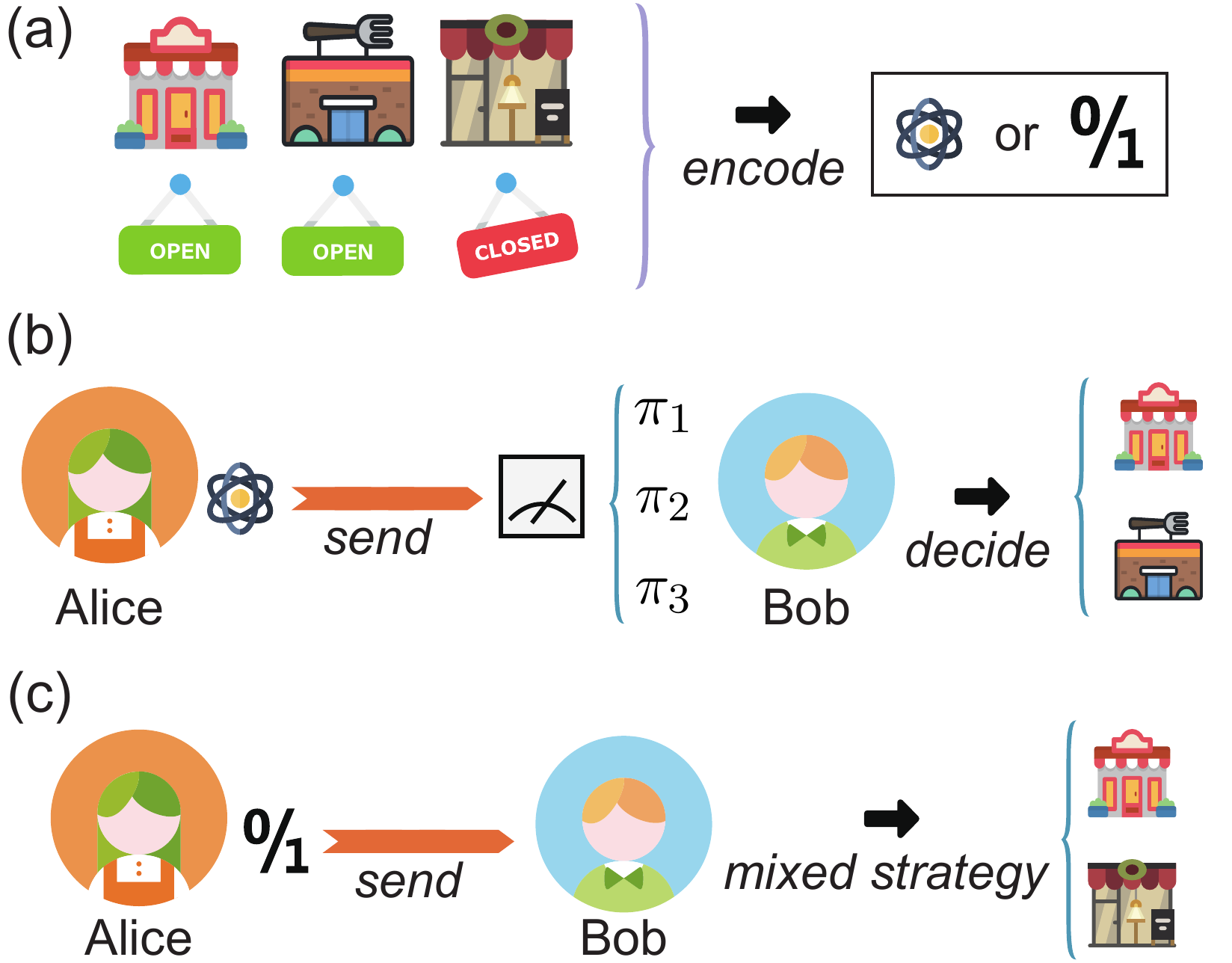}
\end{center}
\caption{\textbf{A three-restaurant game with quantum and classical strategy.}(a) Each day, Alice randomly closes one of three restaurants. Using either a single qubit or a one-bit classical communication, Alice aims to inform Bob in a manner that ensures he never visits a closed restaurant and, over time, visits each restaurant according to a specified probability distribution. (b) In the quantum scenario, Alice sends Bob a qubit state \(\rho_k\) when the \(k^{\text{th}}\) restaurant is closed. Bob then performs a one-shot measurement on this qubit to decide which restaurant to visit. (c) In the classical case, Alice can only transmit a single binary number, and based on this, Bob employs a mixed strategy to choose his restaurant visit.\vspace{-.6cm}}
\label{fig:flow}
\end{figure}

{\it The Restaurant Game.--} The game is played between two parties, Alice and Bob, against an adversary, Eve. Alice manages \( n \) restaurants, one of which (randomly chosen by Alice) remains closed each day. Bob, a friend of Alice, aims to visit one of the open restaurants each day. Alice and Bob communicate privately via a \( d \)-dimensional classical or quantum channel. Bob must pay for Eve's meal if he visits the same restaurant as Eve, so his goal is to minimize the probability of coinciding with Eve while ensuring he never visits a closed restaurant. Eve knows Alice and Bob's strategy but cannot access their private communication. Consequently, Eve’s optimal strategy is to target the restaurant with the highest probability of Bob's visit. Thus, Alice and Bob's optimal strategy should ensure that Bob never visits a closed restaurant and that he visits each open restaurant with equal probability. Let \( p(y|x) \) denote the probability of Bob visiting the \( y \)-th restaurant given that the \( x \)-th restaurant is closed. The optimal strategy for Alice and Bob must satisfy the following conditions:
\begin{align}\label{rgc}
\text{(h1):}~p(y=x|x) = 0,~
\text{(h2):}~\sum_{x \in [n]} p(y|x) p(x) = 1/n,
\end{align}
where \(x,y\in[n] := \{1, \ldots, n\} \). In a recent study involving some of the current authors \cite{Patra2024}, it was demonstrated that the restaurant game with \( n=3 \,, d=2 \) cannot be won using classical strategies, i.e., the correlations in the set \( \mathcal{C}^{3 \times 3}_2 \) do not satisfy the requirements specified in Eq. (\ref{rgc}). However, such correlations are present in \( \mathcal{Q}^{3 \times 3}_2 \), highlighting a quantum advantage. At first glance, this advantage might not seem significant since correlations satisfying Eq. (\ref{rgc}) are included in \( \text{ConvHull}(\mathcal{C}^{3 \times 3}_2) \). However, it is crucial to note that Eve is aware of Alice and Bob's strategy. Consequently, any shared randomness between Alice and Bob becomes ineffective as it allows Eve to tailor her strategy to exploit this information \cite{Self}. This simple game thus underscores the importance of studying the sets \( \mathcal{C}^{3 \times 3}_2 \) and \( \mathcal{Q}^{3 \times 3}_2 \) in their original forms, rather than only considering their convex hulls. Such an approach reveals a quantum advantage in communication games, even when shared randomness is available but not guaranteed to be private. Before establishing this quantum advantage experimentally, please note that the condition (h2) for a generalized restaurant game becomes: $\sum_{x \in [n]} p(y|x) p(x) = \gamma_y,~ \forall~y \in [n]$, where $\gamma_y$ are some given probabilities satisfying $\gamma_y\ge0\,,\sum\gamma_y=1$. We will denote this generalized game as $\mathbb{H}^n(\gamma_1,\cdots,\gamma_n)$, with the shorthand $\mathbb{H}^3(1/3)$ when all the $\gamma_y$'s are equal.

{\it The Qubit Strategy.--} In the game $\mathbb{H}^n(\gamma_1,\cdots,\gamma_n)$ involving a qubit, Alice sends a pure state $\rho_x = \ket{\psi_x}\bra{\psi_x}$ to Bob when the $x^{th}$ restaurant is closed. For decoding, Bob performs the measurement $\mathcal{M} = \{\pi_y := \lambda_y \ket{\psi_y^\perp}\bra{\psi_y^\perp} \mid \sum_{y=1}^n \pi_y = \mathbf{I}_2\}$, and he visits the $y^{th}$ restaurant if the outcome corresponding to the effect $\pi_y$ is observed, where $\ket{\psi_y^\perp}$ denotes the state orthogonal to $\ket{\psi_y}$. Clearly the condition (h1) is satisfied, but satisfying condition (h2) requires carefully chosen encoding and decoding strategies. While only a subset of the three-restaurant games can be perfectly solved with 1-bit of classical communication between Alice and Bob, qubit communication enables perfect strategies for all such games \cite{Patra2024}. Further details on classical and quantum achievable game spaces are discussed in the Supplemental Material. \cite{Supple}. 
\begin{figure}[t]
\begin{center}
\includegraphics[width=0.95\linewidth]{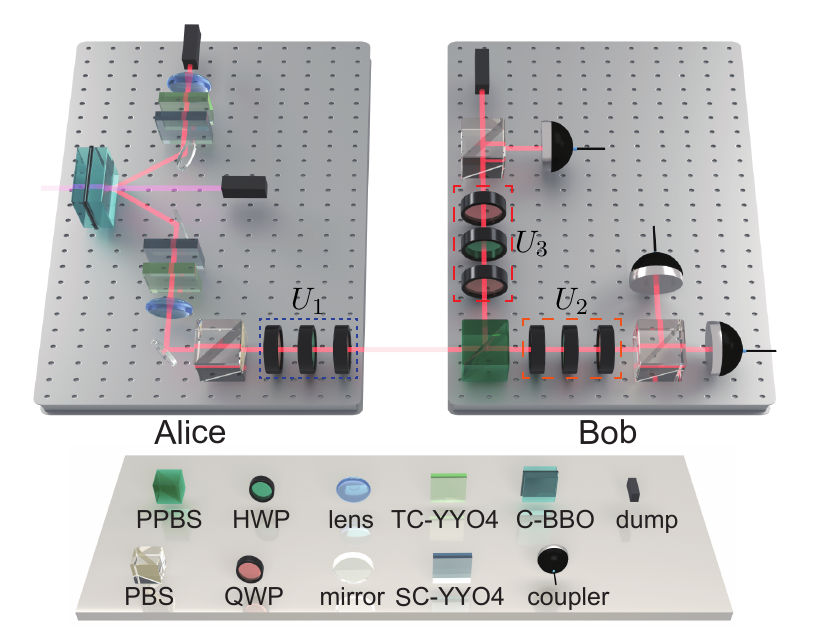}
\end{center}
\caption{\textbf{The experimental setups.} Ultraviolet laser pulses with a central wavelength of 390 nm, a pulse duration of 150 fs, and a repetition rate of 80 MHz are directed through a HWP sandwiched between two BBO crystals to generate a pair of entangled single photons. The entangled photons are then disentangled using a PBS, with post-selection on horizontal polarization, resulting in single photons in the state $\ket{H}$. Alice encodes the restaurant information onto these single photons by applying a HWP sandwiched between two QWPs. Bob employs a specially designed variational triangular polarimeter to implement a series of 3-element POVMs for decoding. Each input photon is detected at one of the three exits of the polarimeter, which determines Bob’s choice of restaurant. Devices: C-BBO (a combination of BBO crystals, HWP, and BBO); SC-YVO4 and TC-YVO4 (YVO4 crystal for spatial and temporal compensation). 
\vspace{-.5cm}}
\label{fig:circuit}
\end{figure}
\begin{figure}[b]
\begin{center}
\includegraphics[width=\linewidth]{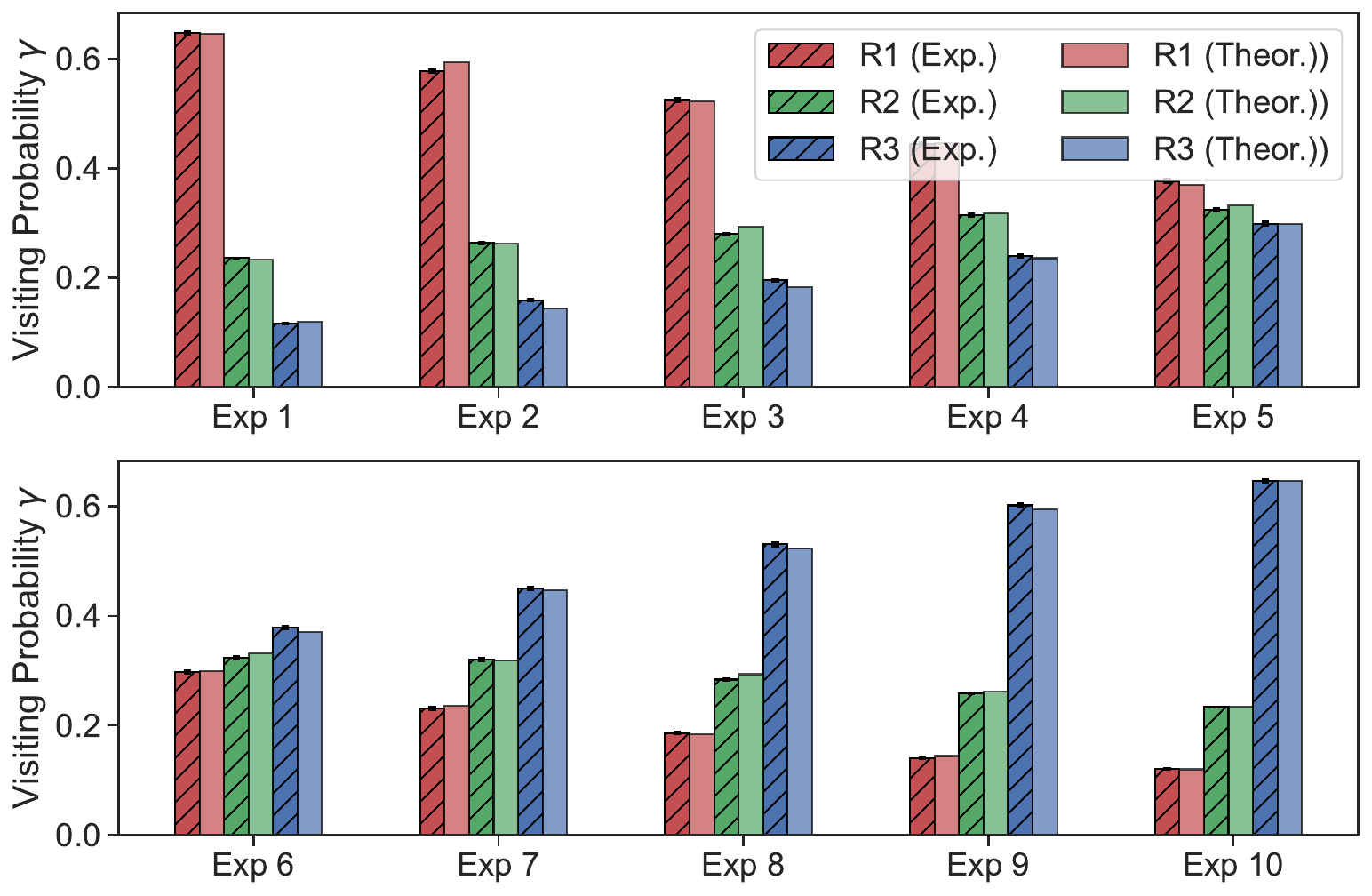}
\end{center}
\caption{\textbf{Bob's probability of visiting restaurants:} Red, green, and blue bars represent the probabilities of visiting restaurants R1, R2, and R3, respectively. The hatched bars denote experimental results, while the solid bars indicate theoretical predictions. The squared statistical overlap between the experimental and theoretical visiting probabilities is calculated as 0.9998(16), demonstrating a high degree of consistency between the experimental data and theoretical predictions.}
\label{fig:bar}
\end{figure}

{\it Experimental Implementation.--} We conduct a series of $\mathbb{H}^3$ restaurant games using a photonic quantum processor. As illustrated in the left panel of Fig.~\ref{fig:circuit}, the experiment begins with the generation of entangled single photons. This is achieved by directing laser pulses through a half-wave plate (HWP) placed between two $\beta$-barium borate (BBO) crystals. The resulting entangled photon pair propagates along two paths. Alice first employs a polarizing beam splitter (PBS) on one of these paths to post-select horizontally polarized photons and disentangle the photon pair, starting with an initial state denoted as $\ket{H}$. She then uses a combination of a HWP and two quarter-wave plates (QWP) to prepare the encoding states. During the $\mathbb{H}^3(\gamma_1,\gamma_2,\gamma_3)$ game, Alice prepares the states $\rho_1, \rho_2, \rho_3$ by adjusting the angles of the wave plates. Given that the QWP-HWP-QWP setup can implement any single-qubit gate for polarization qubits \cite{Simon1990}, Alice can adjust the angles to achieve the desired encoding for any restaurant game. We generally refer to the gate configuration $\text{QWP}(\theta_1)\text{HWP}(\theta_2)\text{QWP}(\theta_3)$ used by Alice as $U_1$.
\begin{figure}[t]
\begin{center}
\includegraphics[width=1\linewidth]{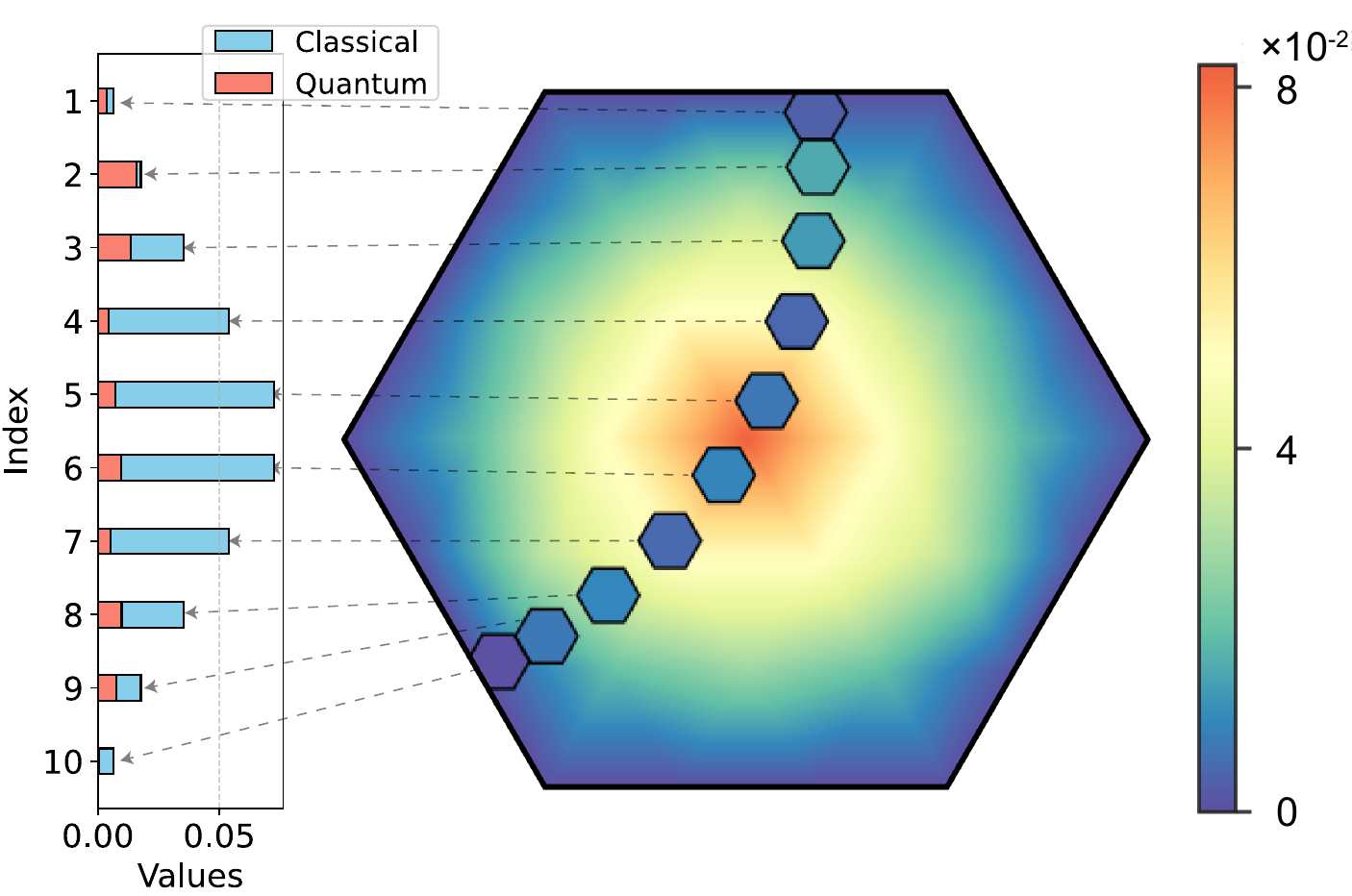}
\end{center}
\caption{\textbf{The winning index values for the quantum and classical strategies.} The large hexagon represents the parameter space for quantumly winnable three-restaurant games, where the theoretical value of $\mathcal{E}_Q$ would be zero. The boundaries of the hexagon correspond to games that are winnable using classical strategies, where $\mathcal{E}_C$ is zero. The interior of the hexagon, which cannot be classically won, is color-coded to show the minimum values of $\mathcal{E}_C$. Inside this parameter space, ten smaller colored hexagons are plotted, with each position indicating the parameters of the game and the colors representing the values of $\mathcal{E}_Q$ obtained from quantum strategies implemented experimentally. For comparison, a bar chart on the left displays the experimentally measured $\mathcal{E}_Q$ for each of the ten games alongside the corresponding optimal classical strategy $\mathcal{E}_C$. Our results demonstrate that in each of the ten experiments, the quantum strategy yields $\mathcal{E}_Q$ values that are significantly lower than the classical $\mathcal{E}_C$ values, underscoring the substantial experimental advantage of quantum strategies.
\vspace{-.5cm}}
\label{fig:tri}
\end{figure}

Upon receiving the encoded photons, Bob needs to apply measurement operators $\mathcal{M}\equiv\{\pi_y:=\lambda_y\ket{\psi^{\perp}_y}\bra{\psi^{\perp}_y}\}_y$, which cannot be implemented using standard polarization measurement devices consisting of a single PBS and wave plates \cite{zhibo2016error}. To address this, we develop a variational triangular polarimeter capable of performing a general 3-element POVM on the photons. As illustrated in the right panel of Fig.~\ref{fig:circuit}, the received photons first encounter a partially polarizing beam splitter (PPBS). This device transmits horizontally polarized photons, while splits a proportion of $(1-f)$ of vertically polarized photons by reflection. The splitting ratio $f$ depends on the structure of the PPBS \cite{alexander2006experimental}, with some designs allowing for immediate tuning \cite{jefferson2018variable}. Each split light then passes through an arbitrary single-qubit gate implemented by a QWP-HWP-QWP combination, followed by a PBS. This setup results in four distinct beams, each of which can be detected independently, thereby defining the measurement operators $\pi_1, \pi_2, \pi_3, \pi_4$. For the $\mathbb{H}^3$ restaurant games, Bob requires only three measurement operators. Consequently, one of the photon detectors is replaced with a light dump, resulting in a 3-element measurement scheme for Bob. By adjusting the angles of the wave plates, the splitting ratio of the PPBS, and incorporating a dice, the polarimeter can effectively implement any 3-element POVM (see Supplemental Material \cite{Supple}).

{\it Results.--} To demonstrate the quantum advantage, we conducted ten $\mathbb{H}^3$ games using the variational triangular polarimeter, with the PPBS having a vertical polarization split ratio of $1/3$. Details of the game settings—including the theoretical visiting probabilities for each restaurant, Alice's encoding states $\rho_k$, the coefficients $\lambda_y$ for Bob's measurement operators, and the rotation angles for the wave plates in $U_1, U_2, U_3$—are provided in the Supplemental Material \cite{Supple}. The experimental visiting probabilities, calculated as the ratio of detected photons at each end of the polarimeter to the total number of detected photons (approximately 4800 per experiment), are presented in Fig.~\ref{fig:bar}. These experimental probabilities closely align with theoretical predictions. The squared statistical overlap \cite{fuchs1996distinguishability} \( F(\vec{\gamma},\vec{p}):= (\sum_y \sqrt{\gamma_y p_y})^2 \), where \( p_y= \sum_{x \in [n]} p(y|x) p(x) \), averaged over the ten games, is calculated as 0.9998(16).

To validate the quantum advantage, we compare the results obtained using quantum and classical strategies. We introduce a quantity $\mathcal{E} := \max\{ k_1 \sum_x p(x|x), ~ k_2 |\gamma_y - p_y|\}_y$, to quantify the deviation of the strategy from ideal conditions (h1) and (h2); here \( p_y = \sum_{x \in [n]} p(y|x) p(x) \), and \( k_1 \) and \( k_2 \) are penalty factors for deviations from conditions (h1) and (h2), respectively. An ideal strategy satisfying both conditions would have \( \mathcal{E} = 0 \). The value of \( \mathcal{E} \) measures the `imperfectness' of the strategy employed. Lower values of \( \mathcal{E} \) indicate closer alignment with the ideal strategy. In this work, we set \( k_1 = 1/3 \) and \( k_2 = 1 \) for simplicity. The method to compute the minimum classical value of \( \mathcal{E} \) for any values of \( k_1 \) and \( k_2 \) is detailed in the Supplemental Material \cite{Supple}. Figure~\ref{fig:tri} illustrates the winning index values for both implemented quantum and classical strategies across each game, plotted within the triangular parameter space of all $\mathbb{H}^3$ games. The average value of \( \mathcal{E}_Q \) for quantum strategies is 0.00120(9), while the average \( \mathcal{E}_C \) for classical strategies is 0.032284. This demonstrates that the quantum strategies consistently outperform the classical ones, highlighting a significant and robust quantum advantage on the realistic device.

{\it Applications.}-- The data storage advantage demonstrated here has significant implications for various quantum applications, including device certification and quantum networks. Nonclassicality, manifested in quantum state through quantum coherence \cite{Streltsov2017}, is crucial in many quantum protocols, such as the BB84 \cite{Bennett84} and BBM92 \cite{Bennett92} cryptographic schemes. Additionally, generalized measurements, like trine-POVM and SIC-POVM, are vital for unambiguous state discrimination tasks that underpin several quantum protocols. However, practical implementation of quantum devices presents challenges: state preparation devices (PDs) often struggle with precise state preparation, and measurement devices (MDs) can be affected by noise. Device-independent and semi-device-independent methods are therefore essential for certifying the nonclassicality of noisy PDs and/or MDs \cite{Scarani2012,Tavakoli2020,Ma2023}. As detailed in the Supplemental Material \cite{Supple}, achieving success in the $\mathbb{H}^3(\gamma_1,\gamma_2,\gamma_3)$ game that exceeds the optimal classical value serves as a semi-device-independent certificate for qubit (PD, MD) pairs.

Moreover, the restaurant game can be viewed as a variant of superdense coding for a single qubit, optimized for efficient data loading and transmission in quantum networks \cite{wang2019characterising}. This approach allows a sender to encode high-dimensional data, represented by normalized positive real vectors $\vec{x}$, using only a single qubit. The receiver can then extract this information adaptively through measurements. The data storage advantage demonstrated on our platform indicates that this method is ready for practical implementation in quantum networks.
\begin{table}[b!]
\centering
\footnotesize
\begin{tabular}{|p{0.05\textwidth}|p{0.4\textwidth}|}
\hline
Year & Results\\
\hline
$1973$ &Holevo's theorem \cite{Holevo1973}: \textbf{In the absence of entanglement,} the maximum mutual information $I(X:Y)$ that can be obtained using a $n$-qubit system and a classical $n$-bit system are equal. This equivalence holds even {\bf in the presence of shared randomness (SR)} as mutual information cannot be increased with SR.\\
\hline
$1992$    &Quantum superdense coding \cite{Bennett1992} (experiments \cite{Mattle1996,Schaetz2004,Williams2017}): \textbf{In the presence of entanglement,} the maximum mutual information $I(X:Y)$ that can be obtained using a $n$-qubit system is larger than a classical $n$-bit system. 
\\
\hline
$2015$&Frenkel-Weiner theorem \cite{Frenkel2015}: \textbf{In the absence of entanglement but the presence of SR,} a $n$-qubit system cannot outperform a classical $n$-bit system in any game.\\
\hline
This work &\textbf{In the absence of SR and entanglement,} a qubit outperforms a c-bit in the Restaurant Game. 
\\
\hline
\end{tabular}
\caption{\textbf{ Contributions of our work in the backdrop of existing no-go results.} X and Y are the random variables associated with the input choice of Alice and output choice of Bob, respectively, in a classical data storage game.} 
\label{tab1l}
\vspace{-.5cm}
\end{table}

{\it Discussions.}-- Our research, as experimentally demonstrated on a photonic quantum processor, establishes that a qubit system can surpass its classical counterpart in storing classical data. This advantage is significant, especially when juxtaposed with the no-go results by Holevo \cite{Holevo1973} and Frenkel-Weiner \cite{Frenkel2015}. In Ref.~\cite{Frenkel2015}, it was shown that the classical information storage capacity of an \(n\)-qubit system is no greater than that of a classical \(n\)-bit system. Notably, the no-go results in Refs.~\cite{Holevo1973, Frenkel2015} consider pre-shared classical correlations or shared randomness (SR) as free resources. While treating SR as a free resource can simplify theoretical analysis, this is not always practical, particularly in communication scenarios. SR, if uncorrelated with potential eavesdroppers, can facilitate private communication and other cryptographic primitives \cite{Ahlswede1993}, and generating such uncorrelated SR is a primary goal in quantum key distribution (QKD) protocols \cite{Vazirani2014}. The role of SR in various tasks has also been explored in several other works \cite{Aumann1987, Babai1997, Canonne2017, Toner2003, Bowles2015, Roy2016, Banik2019, Guha2021, Banerjee2023, Heinosaari2024}. It is important to note that the communication advantage of qubit demonstrated here is fundamentally different  from the advantage known in communication complexity scenario \cite{Buhrman2001, Spekkens2009, Ambainis2019}.

In our experiment, in addition to creating a high-fidelity single-photon source, we develop a variational triangular polarimeter, which enables the implementation of arbitrary 3-element single-qubit POVMs. Remarkably, under certain assumptions, our quantum advantage also provides a method to verify the nonclassicality of both the preparation and measurement devices. Future research could explore further advancements in quantum advantage for classical data storage. Additionally, the experimental techniques developed here have potential applications in simulating complex quantum system dynamics \cite{Liu2019,Elliott2020,Elliott2021,Wu2023}.

\begin{acknowledgments}
{\it Acknowledgements.}-- H.-L. H. is supported by the National Natural Science Foundation of China (Grant No. 12274464), and Natural Science Foundation of Henan (Grant No. 242300421049). EPL acknowledges funding from the QuantERA II Programme that has received funding from the European Union’s Horizon 2020 research and innovation programme under Grant Agreement No 101017733 and the F.R.S-FNRS Pint-Multi programme under Grant Agreement R.8014.21,
from the European Union’s Horizon Europe research and innovation programme under the project ``Quantum Security Networks Partnership'' (QSNP, grant agreement No 101114043), 
from the F.R.S-FNRS through the PDR T.0171.22,
from the FWO and F.R.S.-FNRS under the Excellence of Science (EOS) programme project 40007526,
from the FWO through the BeQuNet SBO project S008323N,
from the Belgian Federal Science Policy through the contract RT/22/BE-QCI and the EU ``BE-QCI'' program. EPL is a FRIA
grantee of the F.R.S.-FNRS. EPL is funded by the European Union. Views and opinions expressed are however those of the author only and do not necessarily reflect those of the European Union. The European Union cannot be held responsible for them.
MA acknowledges the funding supported by the European Union  (Project QURES- GA No. 101153001). MB acknowledges funding from the National Mission in Interdisciplinary Cyber-Physical systems from the Department of Science and Technology through the I-HUB Quantum Technology Foundation (Grant no: I-HUB/PDF/2021-22/008).
\end{acknowledgments}

\end{document}


\title{Supplemental Material for ``Quantum Advantage: A Single Qubit's Experimental Edge in Classical Data Storage''}
\maketitle
 
\section{Theoretical Analysis for the quantum and classical strategies playing 3-restaurant games}
\subsection{Elementary communication scenario}
In the elementary communication scenario involving two parties -- Alice (the sender) and Bob (the receiver) -- Alice receives a classical random variable \( x \in \mathcal{X} \), and Bob's objective is to produce another random variable \( y \in \mathcal{Y} \). This typically yields an \(|\mathcal{X}| \times |\mathcal{Y}|\) stochastic matrix, also known as a channel matrix, with elements \(\{p(y|x)\}\), where \( p(y|x) \geq 0 \) for all \( x, y \) and \( \sum_{y} p(y|x) = 1 \) for each \( x \). The element \( p(y|x) \) represents the conditional probability that Bob produces \( y \) given that Alice received \( x \). Without additional resources, only trivial channel matrices can be achieved. However, when resources are available, a broader range of channel matrices can be realized. These resources are typically categorized into two types:
\begin{figure}[b!]
\centering
\includegraphics[width=0.8\textwidth]{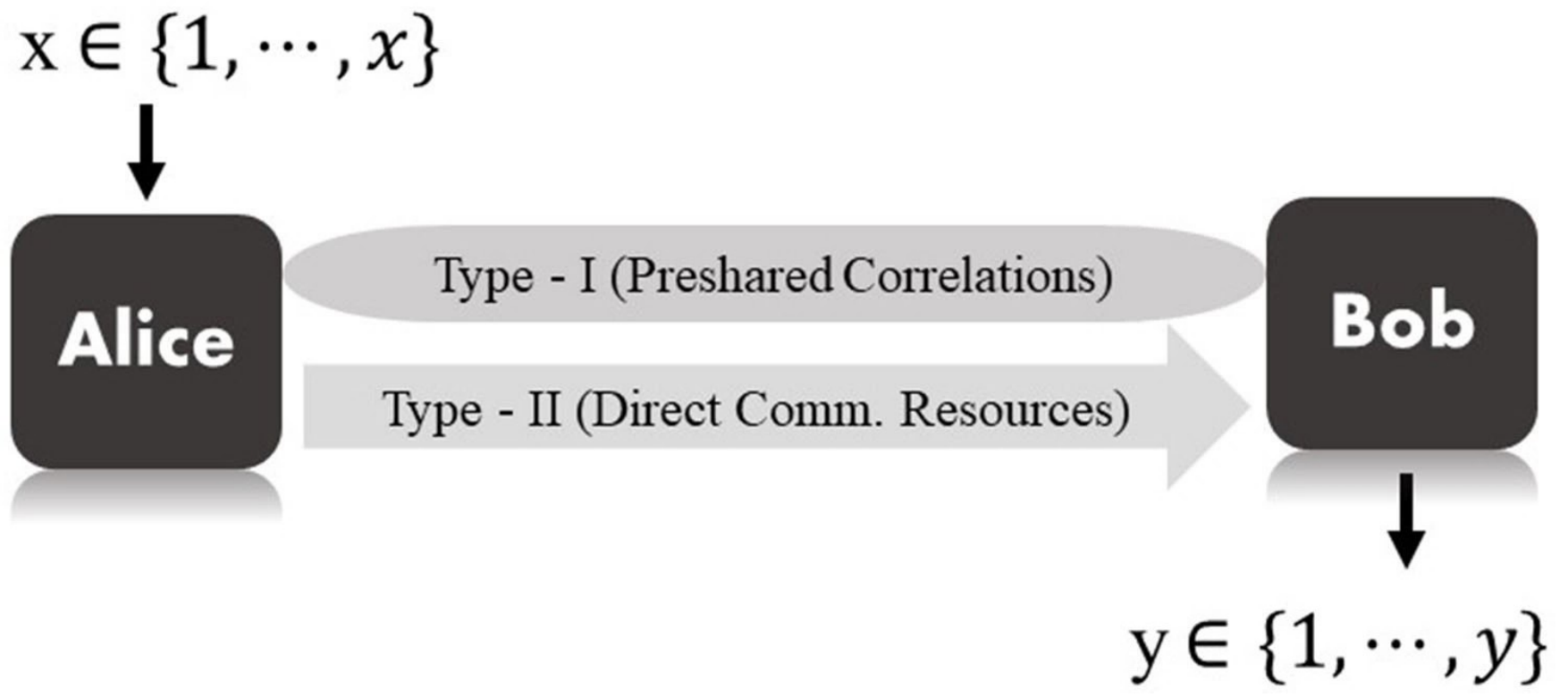}
\caption{ (Color online) Set of channel matrices $\{p(y|x)\}$ that can be simulated by Alice and Bob depends on the available resources available between them. }\label{fig0}
\end{figure}
\begin{itemize}
\item[] {\bf Type-I (Preshared Correlations):} Alice and Bob may have access to prior correlations or correlated systems before initiating the channel simulation. These preshared resources can then be utilized to achieve the desired channel matrix (see Fig.\ref{fig0}). The two main types of preshared correlations are:
\begin{itemize}
\item[C1.] Classical shared randomness, often referred to simply as shared randomness.
\item[C2.] Quantum entanglement, which can lead to intriguing `nonlocal' correlations.
\end{itemize}
\item[] {\bf Type-II (Direct Communication Resource):} Alice can encode her messages into the states of a physical system and send this system directly to Bob. Bob then performs measurements on the received system to obtain the desired outcomes, thereby achieving the target channel matrices. When comparing the communication strengths of different theories, we consider systems with an identical number of distinguishable states. Specifically, we analyze the following types of systems:
\begin{itemize}
\item[D1.] Classical systems: The state of a classical system is described by a probability vector. For a system with \( d \) distinguishable states, the state vectors are in \(\mathbb{R}^d\). A typical example is a classical bit, which has two distinguishable states.
\item[D2.] Quantum systems: The state space of a quantum system, associated with Hilbert space \(\mathcal{H}\), is represented by the set of density operators \(\mathcal{D}(\mathcal{H})\). The number of perfectly distinguishable states corresponds to the dimension of the Hilbert space \(\mathcal{H}\). The quantum analogue of the classical bit is the qubit, whose state space is \(\mathcal{D}(\mathbb{C}^2)\).
\end{itemize}
\end{itemize}
In the following table we list some important results in this elementary communication setup.

\begin{table}[h!]
    \centering
\begin{tabular}{ |p{5cm}|p{5cm}|p{7.5cm}|  }
 \hline 
 \vspace{0.1cm}\centering{\bf Type - I resource} \vspace{0.1cm} &  \vspace{0.1cm}\centering{\bf Type - II Resource}  \vspace{0.1cm}& \vspace{0.1cm} {\bf ~~~~~~~~~~~~~~~~~~~~~~~~~~~~~~~Results} \vspace{0.1cm}\\
 \hline \hline
\vspace{0.3cm}
\centering {Classical Shared Randomness}   & \vspace{0.3cm} \centering{$1$-cbit vs. $1$-qubit}  & The set of channel matrices that can be simulated with $1$-cbit is same as the set simulable with $1$-qubit when classical shared randomness is available between the sender and receiver \cite{Holevo1973, Frenkel2015}.\\
 \hline
 \vspace{0.5cm}
 \centering{Quantum Shared Randomness (Entanglement)}   & \vspace{0.5cm} \centering {$1$-cbit vs. $1$-qubit}& Seminal superdense coding protocol establishes advantage of qubit communication over its classical counterpart in the presence of preshared entanglement \cite{Bennett1992}. Furthermore, entanglement turns out to be more advantageous than classical shared randomness while assisting the cbit \cite{Frenkel2022}.\\
 \hline
\vspace{0.3cm}
 \centering {No Shared Randomness (neither quantum nor classical)}   & \vspace{0.3cm} \centering {$1$-cbit vs. $1$-qubit}  & In the absence of pre-shared correlation, the set of channel matrices achievable through qubit communication is strictly greater than that achievable with classical bit communication \cite{Patra2024}. Our present work experimentally demonstrates this feature.\\
 \hline
\end{tabular}
 \caption{Contributions of the present work in the backdrop of existing no-go results.}
    \label{tab:my_label}
\end{table}

\subsection{Generic Three-Restaurant game: $\mathbb{H}^3(\gamma_1,\gamma_2,\gamma_3)$  }
Three authors from this study, along with their collaborators, have thoroughly investigated the generic \(n\)-Restaurant game, denoted as \(\mathbb{H}^n(\gamma_1, \cdots, \gamma_n)\), in Ref.\cite{Patra2024}. Their analysis covers games solvable with one bit of classical communication as well as those involving qubit communication. For the purposes of the present discussion, we will focus on the detailed exploration of the generic three-restaurant games \(\mathbb{H}^3(\gamma_1, \gamma_2, \gamma_3)\). Alice manages three restaurants, one of which is randomly closed each day. Bob, unaware of which restaurant is closed, aims to visit one of the open restaurants. To assist Bob, Alice can communicate restricted information: either through a perfect classical bit (c-bit) channel or a perfect qubit channel. The objectives are:
\begin{itemize}
\item[(h1)] Bob must avoid visiting the closed restaurant.
\item[(h2)] Bob's probability of visiting restaurant \( i \) should be \(\gamma_i\), where \(\gamma_i \geq 0\) for all \( i \in \{1, 2, 3\}\) and \(\sum_{i=1}^3 \gamma_i = 1\).
\end{itemize}
This scenario can be represented using the `visit' matrix:
\begin{align}\label{visit}
\mathbb{V}\equiv\begin{blockarray}{cccc}
& 1_b & 2_b & 3_b \\
\begin{block}{c(ccc)}
1_c~~ & p(1_b|1_c) & p(2_b|1_c) & p(3_b|1_c) \\
2_c~~ & p(1_b|2_c) & p(2_b|2_c) & p(3_b|2_c) \\
3_c~~ & p(1_b|3_c) & p(2_b|3_c) & p(3_b|3_c) \\
\end{block}
\end{blockarray}  
\end{align}
Here, \( p(i_b|j_c) \) denotes the probability that Bob visits the \( i \)-th restaurant (indicated by the subscript \( b \)) given that the \( j \)-th restaurant is closed (indicated by the subscript \( c \)). The visit matrix \(\mathbb{V}\) is a stochastic matrix where \(\sum_{i_b=1}^3 p(i_b|j_c) = 1\) for all \( j_c \). The condition (h1) implies that:
\begin{align}
p(i_b|i_c) = 0, \quad \forall i \in \{1, 2, 3\},
\end{align}
meaning all diagonal entries of \(\mathbb{V}\) must be zero. Bob's probability of visiting the \( i \)-th restaurant is obtained by summing the entries in the \( i \)-th column of the matrix \(\mathbb{V}\). Thus, condition (h2) is expressed as:
\begin{align}\label{gamma}
p(i_b) = \sum_{j=1}^3 p(i_b|j_c) p(j_c) = \frac{1}{3} \sum_{j=1}^3 p(i_b|j_c) = \gamma_i.
\end{align}
Here, \( p(j_c) = 1/3 \) for all \( j \in \{1, 2, 3\} \), assuming uniform closing probabilities, and the final equality follows from condition (h2). To explore the allowed game space \(\mathbb{H}^3(\gamma_1, \gamma_2, \gamma_3)\) under condition (h1) without any communication constraints, we first identify the extreme points within the vector \(\vec{\gamma} = (\gamma_1, \gamma_2, \gamma_3)^T\). It is evident that the probability vector \(\vec{\gamma}\) cannot have two zero entries simultaneously; at most, one zero is possible. We determine the extreme visit matrices \(\mathbb{V}_e\) that satisfy condition \( \mathrm{h1} \), leading to extreme \(\vec{\gamma}\) vectors. The extreme visit matrices are:
\begin{align}
\left\{\!\begin{aligned}
\mathbb{V}_{e}^{1}\equiv\begin{blockarray}{ccc}
\begin{block}{(ccc)}
 0 & 1 & 0 \\
 1 & 0 & 0 \\
1 & 0 & 0 \\
\end{block}
\end{blockarray} ~~~;~~
\mathbb{V}_{e}^{2}\equiv\begin{blockarray}{ccc}
\begin{block}{(ccc)}
 0 & 1 & 0 \\
 1 & 0 & 0 \\
0 & 1 & 0 \\
\end{block}
\end{blockarray}  ~~~;~~
\mathbb{V}_{e}^{3}\equiv\begin{blockarray}{ccc}
\begin{block}{(ccc)}
 0 & 1 & 0 \\
 0 & 0 & 1 \\
0 & 1 & 0 \\
\end{block}
\end{blockarray} ~~~;~~
\mathbb{V}_{e}^{4}\equiv\begin{blockarray}{ccc}
\begin{block}{(ccc)}
 0 & 0 & 1 \\
 0 & 0 & 1 \\
0 & 1 & 0 \\
\end{block}
\end{blockarray}  ~\\
\mathbb{V}_{e}^{5}\equiv\begin{blockarray}{ccc}
\begin{block}{(ccc)}
 0 & 0 & 1 \\
 1 & 0 & 0 \\
1 & 0 & 0 \\
\end{block}
\end{blockarray} ~~~;~~
\mathbb{V}_{e}^{6}\equiv\begin{blockarray}{ccc}
\begin{block}{(ccc)}
 0 & 0 & 1 \\
 0 & 0 & 1 \\
1 & 0 & 0 \\
\end{block}
\end{blockarray}  ~~~;~~
\mathbb{V}_{e}^{7}\equiv\begin{blockarray}{ccc}
\begin{block}{(ccc)}
 0 & 0 & 1 \\
 1 & 0 & 0 \\
0 & 1 & 0 \\
\end{block}
\end{blockarray}~~~;~~
\mathbb{V}_{e}^{8}\equiv\begin{blockarray}{ccc}
\begin{block}{(ccc)}
 0 & 1 & 0 \\
 0 & 0 & 1 \\
1 & 0 & 0 \\
\end{block}
\end{blockarray} ~
\end{aligned}\right\}.
\end{align}
From these extreme visit matrices and Equation (\ref{gamma}), the six extreme probability vectors \(\vec{\gamma}_e^1\) through \(\vec{\gamma}_e^6\) are:
\begin{align}
\left\{\!\begin{aligned}
\vec{\gamma}_e^1 = \begin{blockarray}{c}
\begin{block}{(c)}
 2/3 \\
 1/3 \\
 0 \\
\end{block}
\end{blockarray} ~~;~
\vec{\gamma}_e^2 = \begin{blockarray}{c}
\begin{block}{(c)}
 1/3 \\
 2/3 \\
 0 \\
\end{block}
\end{blockarray}  ~~;~
\vec{\gamma}_e^3 = \begin{blockarray}{c}
\begin{block}{(c)}
 0 \\
 2/3 \\
 1/3 \\
\end{block}
\end{blockarray}  ~~;~
\vec{\gamma}_e^4 = \begin{blockarray}{c}
\begin{block}{(c)}
 0 \\
 1/3 \\
 2/3 \\
\end{block}
\end{blockarray} ~~;~
\vec{\gamma}_e^5 = \begin{blockarray}{c}
\begin{block}{(c)}
 2/3 \\
 0 \\
 1/3 \\
\end{block}
\end{blockarray}  ~~;~
\vec{\gamma}_e^6 = \begin{blockarray}{c}
\begin{block}{(c)}
 1/3 \\
 0 \\
 2/3 \\
\end{block}
\end{blockarray} ~ 
\end{aligned}\right\}.
\end{align}

Matrices \(\mathbb{V}_e^7\) and \(\mathbb{V}_e^8\) result in a uniform probability distribution \(\vec{\gamma} = (1/3, 1/3, 1/3)^T\). This vector is not extreme, as it can be expressed as a convex combination of the other extreme probability vectors identified. The allowable game space \(\mathbb{H}^3(\gamma_1, \gamma_2, \gamma_3)\) is therefore represented by the convex hull of the extreme probability vectors \(\{\vec{\gamma}_e^i\}_{i=1}^6\). However, in scenarios where communication is restricted or where shared randomness is not permitted, the permissible game space may be more limited than this set. In the next subsection, we examine the achievable game spaces \(\mathbb{H}^3(\gamma_1, \gamma_2, \gamma_3)\) for Alice and Bob when using a perfect 1-bit classical channel or a perfect qubit channel. It is crucial to note that in these analyses, we assume the presence of pre-shared correlations as a resource and investigate the scenarios where Alice and Bob operate without such pre-shared correlations.

\subsubsection{Classical achievable game space for three - Restaurant game}
In the context of the three-restaurant game, when Alice is restricted to communicating only 1 bit, the most general strategies can be described as follows:
\begin{itemize}
\item[--] If the \(i\)-th restaurant is closed, Alice will communicate a bit \(0\) with probability \(\alpha_i\) and a bit \(1\) with probability \(1 - \alpha_i\).
\item [--] Upon receiving the signal \(0\), Bob decides his choice by flipping a three-faced coin with probabilities \(\{r_i\}_{i=1}^3\), leading him to visit the \(i\)-th restaurant. Conversely, upon receiving the signal \(1\), Bob uses a different three-faced coin with probabilities \(\{q_i\}_{i=1}^3\), determining his choice of the \(i\)-th restaurant.
\end{itemize}
These strategies represented as a visit matrix \(\mathbb{V}\) reads as:
\begin{align}\label{visit3}
\mathbb{V}\equiv\begin{blockarray}{cccc}
& 1_b & 2_b & 3_b \\
\begin{block}{c(ccc)}
1_c~~ & \alpha_1r_1+(1-\alpha_1)q_1 & \alpha_1r_2+(1-\alpha_1)q_2 & \alpha_1r_3+(1-\alpha_1)q_3 \\
2_c~~ & \alpha_2r_1+(1-\alpha_2)q_1 & \alpha_2r_2+(1-\alpha_2)q_2 & \alpha_2r_3+(1-\alpha_2)q_3 \\
3_c~~ & \alpha_3r_1+(1-\alpha_3)q_1 & \alpha_3r_2+(1-\alpha_3)q_2 & \alpha_3r_3+(1-\alpha_3)q_3. \\
\end{block}
\end{blockarray}  
\end{align}
To satisfy condition (h1), which requires that every diagonal term of the visit matrix \(\mathbb{V}\) be zero, i.e., $p(i_b | i_c) = 0 \quad \text{for all } i \in \{1, 2, 3\}$. Hence, the equation ensuring this condition becomes:
\begin{align}
\alpha_ir_i+(1-\alpha_i)q_i=0,~~\forall~i=1,2,3.
\end{align}
Satisfying these requirements following three distinct strategies can be employed:
\begin{itemize}
\item[(i)] set $\alpha_{i}=0$ and $q_i=0$,
\item[(ii)] set $\alpha_i=1$ and $r_i=0$,
\item[(iii)] set $r_i = 0$ and $q_i =0$.
\end{itemize}

To determine the game space \( \mathbb{H}^3(\gamma_1, \gamma_2, \gamma_3) \) given the constraints on the probability distribution, we explore two main cases:\\
{\bf Case 1:} [At most one \(\gamma_i\) is Zero]    When one of the probabilities \(\gamma_i\) is zero, say \(\gamma_1 = 0\), the strategies Alice and Bob can use are described as follows:
\begin{itemize}
\item[]{\it Alice's Strategy:}
\item[--] If Restaurant \(2\) is closed, Alice communicates \(0\).
\item[--] If Restaurant \(3\) is closed, Alice communicates \(1\).
\item[--] If Restaurant \(1\) is closed, Alice communicates \(0\) with probability \(\alpha\) and \(1\) with probability \(1 - \alpha\).
\item[]{\it Bob's Strategy:}
\item[--] Upon receiving \(0\), Bob decides to visit Restaurant \(2\).
\item[--] Upon receiving \(1\), Bob decides to visit Restaurant \(3\)
\end{itemize}
This strategy ensures that Bob never visits a closed restaurant. The resulting probabilities are:
\begin{align*}
\gamma_1 = 0, \quad \gamma_2 = \frac{1 + \alpha}{3}, \quad \gamma_3 = \frac{2 - \alpha}{3}.    
\end{align*}
By choosing \(\alpha\) appropriately from the interval \([0, 1]\), Alice and Bob can achieve any valid game scenario where \(\gamma_1 = 0\). Similar strategies can be applied for cases where \(\gamma_2 = 0\) or \(\gamma_3 = 0\).\\
{\bf Case 2:} [All \(\gamma_i > 0\)] In this case, Alice and Bob must use strategies where \(r_i \neq 0\) and \(q_i \neq 0\) for all \(i\). To satisfy the condition that every diagonal entry of the visit matrix \(\mathbb{V}\) is zero, we partition the set \(\{1, 2, 3\}\) into two subsets \(X\) and \(Y\) where: for the restaurants in subset \(X\), Alice follows strategy \((i)\), and for the restaurants in subset \(Y\), Alice follows strategy \((ii)\). The possible partitions and their corresponding strategies are:
\begin{itemize}
\item[--] Partition \(R_1\): \(r_2 = r_3 = q_1 = 0\), \& \(r_1 = 1\), \& \(q_2 + q_3 = 1\); resulting probabilities: \(\gamma_1 = \frac{2}{3}\), \(\gamma_2 + \gamma_3 = \frac{1}{3}\).
\item[--] Partition \(R_2\): \(r_1 = r_3 = q_2 = 0\), \& \(r_2 = 1\), \& \(q_1 + q_3 = 1\); resulting probabilities: \(\gamma_2 = \frac{2}{3}\), \(\gamma_1 + \gamma_3 = \frac{1}{3}\).
\item[--] Partition \(R_3\): \(r_1 = r_2 = q_3 = 0\), \& \(r_3 = 1\), \& \(q_1 + q_2 = 1\); resulting probabilities: \(\gamma_3 = \frac{2}{3}\), \(\gamma_1 + \gamma_2 = \frac{1}{3}\). 
\end{itemize}
By inverting Alice's encoding (i.e., swapping \(0 \leftrightarrow 1\)), these scenarios cover all possible partitions. Thus, the game is perfectly winnable with a classical mixed strategy if and only if one of the \(\gamma_i\) values is \( \frac{2}{3} \) or \(0\). Therefore, the achievable game space \(\mathbb{H}^3(\gamma_1, \gamma_2, \gamma_3)\) with a classical mixed strategy is restricted to cases where one of the probabilities is \( \frac{2}{3} \) or \(0\). This condition highlights the constraints on the game space in classical strategies, as visualized in Figure \ref{fig:S}.
\begin{center}
\begin{table}[t!]
\begin{tabular}{ |c||c|c|c|c|c|c|  }
\hline
~~Set~~ & $R_1$ & $R_2$ & $R_3$ & $R_4$ & $R_5$ & $R_6$\\
\hline\hline
X & $\{1\}$ & $\{2\}$ & $\{3\}$ & $~\{2,3\}~$ & $~\{1,3\}~$ & $~\{1,2\}~$\\
\hline
Y & $~\{2,3\}~$ & $~\{1,3\}~$ & $~\{1,2\}~$ & $\{1\}$ & $\{2\}$ & $\{3\}$\\
\hline
\end{tabular}
\caption{Six different partitioning of three Restaurants into two nonempty disjoint sets.}
\label{table2}
\end{table}
\end{center}
\vspace{-.7cm}

\begin{figure}[t]
\centering
\includegraphics[width=0.5\textwidth]{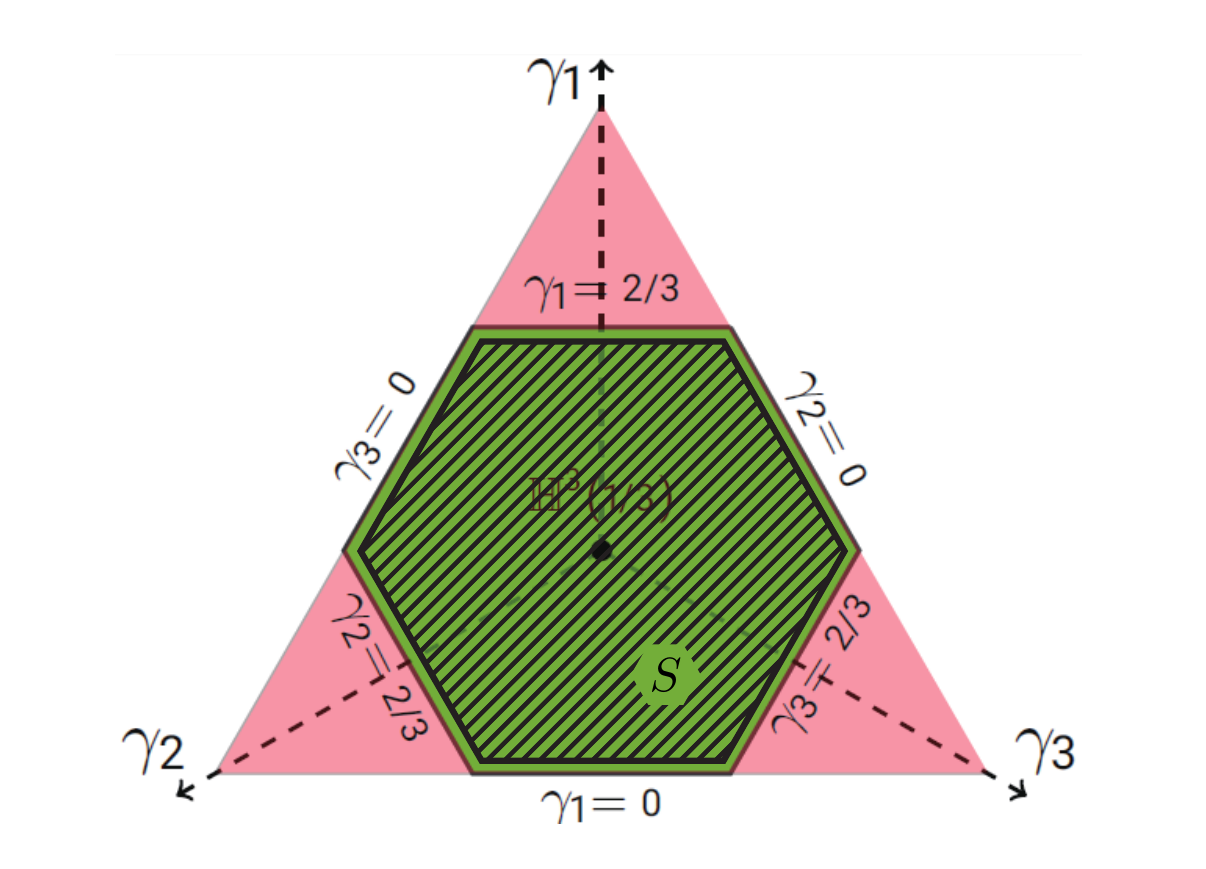}
\caption{(color online) The parameter space, depicted on the $\gamma_1+\gamma_2+\gamma_3=1 $ plane, represents the set of games $ \mathbb{H}^3(\gamma_1,\gamma_2,\gamma_3) $ . In this depiction, the regions shaded in orange denote unphysical games, as they violate condition  (h1). Conversely, the shaded green region forms a polytope encompassing all allowed games. The boundaries of this green polytope, specifically $  \gamma_i=0 $  and $ \frac{2}{3} $  for $  i\in\{1,2,3\} $, represent the exclusive set of games that can be won with only 1 classical bit communication. On contrary whole game space can be achievable by 1 qubit communication. The closed set $S$ that robust quantum advantage exists, inside the parameter-space ($\gamma_1+\gamma_2+\gamma_3=1$ plane) of all the games $\mathbb{H}^3(\gamma_1,\gamma_2,\gamma_3)$.}\label{fig:S}
\end{figure}

\subsubsection{Quantum achievable game space}
In this general case, to satisfy condition (h1), Alice must select pure states for encoding. She communicates (through a noiseless qubit channel) the state $\ket{\psi_i}$ to Bob when the $i^{th}$ restaurant is closed. In order to fulfill condition (h1), Bob needs to perform a decoding measurement represented by $\mathcal{M}=\left\{\alpha_i\ketbra{\psi_i^{\perp}}{\psi_i^{\perp}}~|~\alpha_i>0~\&~\sum_i\alpha_i=2\right\}_{i=1}^3$. He then visits the $i^{th}$ restaurant if the $i^{th}$ effect clicks. To meet this requirement, the completely mixed state must reside within the triangle formed by Bloch vectors corresponding to the encoding states $\{\ket{\psi_i}\}_{k=1}^3$. Without loss of any generality, Alice can choose her encodings as $\psi_1=(0,0,1)^{\mathrm{T}},~\psi_2=(-\sin\theta_2,0,\cos\theta_2)^{\mathrm{T}}$, and $\psi_3=(\sin\theta_3,0,\cos\theta_3)^{\mathrm{T}}$; where $\psi_i$ is the Bloch vector of the state $\ket{\psi_i}$, and $\theta_2,\theta_3\in[0,\pi]$ are the polar angles for the corresponding Bloch vectors (see Fig.\ref{fig6}). Note that $\theta_2+\theta_3$ cannot be less than $\pi$, as this configuration would not constitute a valid measurement.  Accordingly, the conditions set forth by (h1) can be expressed as follows:

\begin{subequations}
\begin{align}
\alpha_1+\alpha_2+\alpha_3=2,~~
\alpha_1+\alpha_2\cos\theta_2+\alpha_3\cos\theta_3=0,~~
-\alpha_2\sin\theta_2+\alpha_3\sin\theta_3=0.
\end{align}
\end{subequations}
These equations subsequently lead to:
\begin{subequations}
\begin{align}
\label{eq16}
\alpha_1=\frac{2\sin(\theta_2+\theta_3)}{\sin(\theta_2+\theta_3)-\sin\theta_2-\sin\theta_3},~~
\alpha_2=\frac{-2\sin\theta_3}{\sin(\theta_2+\theta_3)-\sin\theta_2-\sin\theta_3},~~
\alpha_3=\frac{-2\sin\theta_2}{\sin(\theta_2+\theta_3)-\sin\theta_2-\sin\theta_3},
\end{align}
\end{subequations}
and accordingly, we have,
\begin{align}
\label{eq17}
\gamma_1&= \frac{1}{3}(p(1|2)+p(1|3))= \frac{1}{3}\Tr[\alpha_1\ketbra{\psi_1^{\perp}}{\psi_1^{\perp}}\left(\ketbra{\psi_2}{\psi_2}+\ketbra{\psi_3}{\psi_3}\right)]= \frac{1}{3}\frac{\sin(\theta_2+\theta_3)(2-\cos\theta_2-\cos\theta_3)}{(\sin(\theta_2+\theta_3)-\sin\theta_2-\sin\theta_3)}.
\end{align}

\begin{figure} [h!]
\centering
\includegraphics[width=0.4\textwidth]{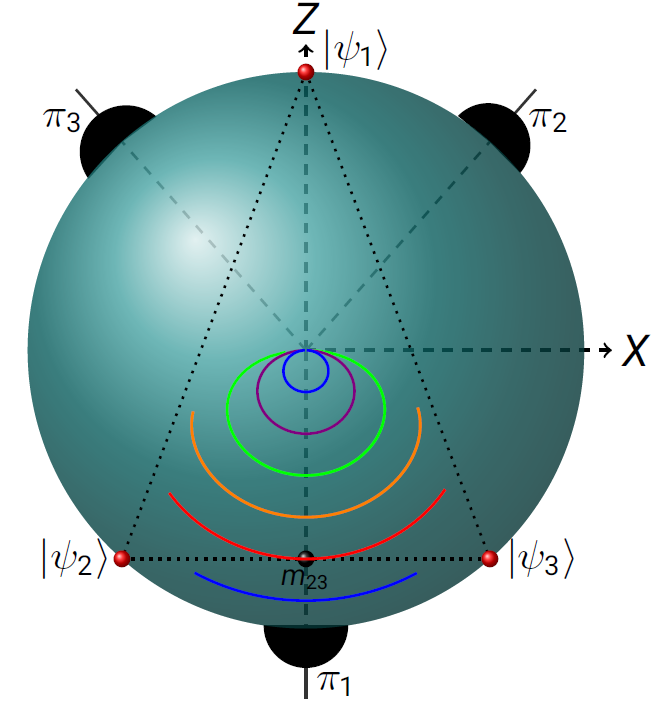}
\caption{(Color online) This figure shows locus of the midpoints $m_{23}$ (denoted by black dot) having constant $\gamma_1$. Once $\gamma_1$ is fixed, $m_{23}$ completely specifies the value of $\gamma_2$ and $\gamma_3$. Thus it is sufficient to plot constant $\gamma_1$ curves. Here, $\gamma_1=0.6,~0.5,~0.4,~0.3,~0.2, \&~0.1$ curves are plotted.The black dotted triangle, encapsulating the encoding states $\ket{\psi_i}$, illustrates a concrete strategy for the game $\mathbb{H}^3\left(0.5,0.25,0.25\right)$. As the black dot moves on the blue curve we get strategies for the games of the form $\mathbb{H}^3\left(0.5,\frac{p}{2},\frac{1-p}{2}\right)$, with $p\in[0,1]$. The leftmost point on the red curve corresponds to the game $\mathbb{H}^3\left(0.5,0,0.5\right)$ while the rightmost point corresponds to $\mathbb{H}^3\left(0.5,0.5,0\right)$.}\label{fig6}
\end{figure}

The encoding described above can be uniquely determined by setting the state $\ket{\psi_1}=\ket{0}$ and fixing the midpoint $m_{23}$ of the line segment connecting the Bloch vectors of $\ket{\psi_2}$ and $\ket{\psi_3}$ (see Fig.\ref{fig6}). This arises from the mathematical property that for any point within a great circle (except the center), there exists a unique chord with that point as its midpoint. By selecting the encoding state $\ket{\psi_1}=\ket{0}$ when the first restaurant is closed, the entire encoding is effectively characterized solely by the position of the midpoint $m_{23}$, with the Bloch vector corresponding to Restaurant $2$ positioned to the left.

Subsequently, Eq.(\ref{eq17}) can be utilized to plot the locus of midpoints $m_{23}$ with constant $\gamma_1$. Such a graphical representation is illustrated in Fig.\ref{fig6}. Consequently, we deduce that all games $\mathbb{H}^3(\gamma_1,\gamma_2,\gamma_3)$ can be won perfectly through some quantum strategy.

\subsection{Robust quantum advantage}

We show there is a robust quantum advantage on a closed set of $\mathbb{H}^3$ restaurant games, demonstrated by Fig.~\ref{fig:S}.

\begin{prop}
For an arbitrary game in the marked hexagon $S$ demonstrated in Fig.~\ref{fig:S}, there exist a noise level $\epsilon$, under which the quantum strategy implemented on a noisy quantum processor performs better than the best classical strategy. Mathematically,
\begin{align}
    \forall\ \mathbb{H}^3(\gamma_1,\gamma_2,\gamma_3)\in S,\ \exists\ \epsilon, \mathcal{E}_\text{noisy quantum}(\epsilon,\gamma_1,\gamma_2,\gamma_3)<\mathcal{E}_\text{classical}(\gamma_1,\gamma_2,\gamma_3),
\end{align}
in which $\mathcal{E}_\text{noisy quantum}(\epsilon,\gamma_1,\gamma_2,\gamma_3)$ is the quality index value of the quantum strategy implemented on a noisy quantum processor with noise level $\epsilon$, $\mathcal{E}_\text{classical}(\gamma_1,\gamma_2,\gamma_3)$ is the quality index value of an arbitrary classical strategy, playing the game $\mathbb{H}^3(\gamma_1,\gamma_2,\gamma_3)$.
\end{prop}

Without loss of generality, we assume
\begin{align}
    \lim_{\epsilon\rightarrow 0}\mathcal{E}_\text{noisy quantum}(\epsilon,\gamma_1,\gamma_2,\gamma_3)&=\mathcal{E}_\text{ideal quantum}(\gamma_1,\gamma_2,\gamma_3)\\
    &=0.
\end{align}

A general classical strategy without pre-shared information must be a mixed strategy, such that 
\begin{align}\label{eqn:E_pa_pb}
    \mathcal{E}(\vec{p}^A,\vec{p}^B,\gamma_1,\gamma_2,\gamma_3)=\max_j \left\{k_1 \sum_i \sum_j p^A(j|i_c) p^B(i_b|j), ~k_2\left|\gamma_j-\frac{1}{n}\sum_i \sum_k p^A(k|i_c) p^B(j_b|k)\right| \right\}.
\end{align}
in which $p^A(j|i_c)$ is the probability that Alice send $j$ when $i^\text{th}$ restaurant is closed, $p^B(k_b|j)$ is the probability that Bob visits $k^\text{th}$ restaurant on receiving $j$. The lowest $\mathcal{E}_\text{classical}$ value is achieved by optimizing the the strategy-related probabilities $\vec{p}^A,\vec{p}^B$, as
\begin{align}\label{eqn:E_best}
    \mathcal{E}_\text{best classical}(\gamma_1,\gamma_2,\gamma_3)&=\min_{\vec{p}^A,\vec{p}^B} \mathcal{E}(\vec{p}^A,\vec{p}^B,\gamma_1,\gamma_2,\gamma_3),
\end{align}

Form the definition in Eqn.~(\ref{eqn:E_pa_pb}), we can easily find the function $\mathcal{E}(\vec{p}^A,\vec{p}^B,\gamma_1,\gamma_2,\gamma_3)$ is continuous and larger than zero. Therefore, the function $\mathcal{E}_\text{best classical}(\gamma_1,\gamma_2,\gamma_3)$ in Eqn.~(\ref{eqn:E_best}) is well-defined. Since $S$ is closed, it suffices to prove that $\mathcal{E}_\text{best classical}(\gamma_1,\gamma_2,\gamma_3)$ is continuous on $S$. 

We generally prove the following proposition.

\begin{prop}
Given a continuous multivariate function $f(\vec{x},\vec{y})$, supposing $g(\vec{y})=\min_{\vec{x}} f(\vec{x},\vec{y})$, then $g(\vec{y})$ is continuous.
\end{prop}
\begin{proof}
    
We prove by contradiction. Supposing $g(\vec{y})$ is not continuous on the point $g(\vec{y}_1)$, then
\begin{align}
    \exists \delta>0, \forall \gamma>0 \text{and} \|\vec{y}_1-\vec{y}_2\|&<\gamma,\\
    \left|\min_{\vec{x}} f(\vec{x},\vec{y_1})-\min_{\vec{x}} f(\vec{x},\vec{y_2})\right|&>\delta.
\end{align}

Without loss of generality, supposing $\min_{\vec{x}} f(\vec{x},\vec{y_1})>\min_{\vec{x}} f(\vec{x},\vec{y_2})$, then
\begin{align}
    \min_{\vec{x}} f(\vec{x},\vec{y_1})>\min_{\vec{x}} f(\vec{x},\vec{y_2})+\delta.
\end{align}

Since $f(\vec{x},\vec{y})$ is continuous, there exist $\vec{x}_1,\vec{x}_2$, such that
\begin{align}
    \min_{\vec{x}} f(\vec{x},\vec{y_1})=f(\vec{x_1},\vec{y_1}),\\
    \min_{\vec{x}} f(\vec{x},\vec{y_2})=f(\vec{x_2},\vec{y_2}).
\end{align}

Besides,
\begin{align}
    \exists \gamma_2, \forall \|(\vec{x_2},\vec{y_1})-(\vec{x_2},\vec{y_2})\|&<\gamma_2,\\
    |f(\vec{x_2},\vec{y_1})-f(\vec{x_2},\vec{y_2})|&<\frac{\delta}{2}.
\end{align}

Then,
\begin{align}
    f(\vec{x_2},\vec{y_1})&<f(\vec{x_2},\vec{y_2})+\frac{\delta}{2},\\
    &<f(\vec{x_1},\vec{y_1}),
\end{align}
which is contradictory to the fact that $\min_{\vec{x}} f(\vec{x},\vec{y_1})=f(\vec{x_1},\vec{y_1})$.

\end{proof}

\section{Experimental implementation details}
\subsection{The single-qubit photon source}
Laser pulses with a central wavelength of 390 nm, pulse duration of 150 fs, and repetition rate of 80 MHz pass through a half-wave plate (HWP) sandwiched by two $\beta$-barium borate (BBO) crystals, where pairs of entangled photons are generated by spontaneous parametric down-conversion, flying towards two directions. A polarized beam splitter (PBS) is placed on one of the paths to disentangle the photons and post-select horizontal polorization, generating single photons of state $\ket{H}$.




\subsection{Arbitrary 3-element POVM}

We introduce the functionality of the designed variational triangular polarimeter and show it can be extended to implement an arbitrary 3-element single-qubit POVM. The photons that Bob receives first pass through a partially-polarizing beam splitter (PPBS), which reflects a proportion of $1-x$ of the vertically polarized photons and transmits the remaining vertically polarized photons and all horizontally polarized photons. 
The PPBS introduces a new path qubit for the photons. Suppose the received state is $\ket{\psi_\text{in}}=\alpha\ket{0}+\beta\ket{1}$, with the following six wave plates on the two ends of the PPBS, the output state is generally 
\begin{equation}
\ket{\psi_\text{out}}=U_2\left(\alpha\ket{H}+\sqrt{x}\beta\ket{V}\right)\ket{T}+\sqrt{1-x}\beta U_3\ket{V}\ket{R},
\end{equation}
in which $\ket{T}$ represents the state in the transimission path, $\ket{R}$ represents the state in the reflection path, $U_2,U_3$ are two arbitrary single-qubit gates implemented by the wave plates. After the wave plates, Bob uses polarizing beam splitters (PBS) and single-photon detectors to measure the photons on the basis of $\{\ket{H},\ket{V}\}$ and $\{\ket{T},\ket{R}\}$. Denoting $\ket{H},\ket{T}$ as $\ket{0}$, and $\ket{V},\ket{R}$ as $\ket{1}$, the polarimeter defines four measurement operators $\pi_0,\pi_1,\pi_2,\pi_3$ by 
\begin{align}
\tr{[\pi_i\ket{\psi_\text{in}}\bra{\psi_\text{in}}]}=|\braket{i|\psi_\text{out}}|^2,\quad i=0,1,2,3.
\end{align}

Supposing
\begin{align}
U_2&=\begin{pmatrix}
u_{2,00}&u_{2,01}\\
u_{2,10}&u_{2,11}
\end{pmatrix},~~\&~~U_3=\begin{pmatrix}
u_{3,00}&u_{3,01}\\
u_{3,10}&u_{3,11}
\end{pmatrix},
\end{align}
we have
\begin{align}
\pi_0&=\begin{pmatrix}
|u_{2,00}|^2&\sqrt{x}\overline{u_{2,00}}u_{2,01}\\
\sqrt{x}u_{2,00}\overline{u_{2,01}}&x|u_{2,01}|^2
\end{pmatrix},\label{sm:eqn:pi0}\\
\pi_1&=\begin{pmatrix}
|u_{2,10}|^2&\sqrt{x}\overline{u_{2,10}}u_{2,11}\\
\sqrt{x}u_{2,10}\overline{u_{2,11}}&x|u_{2,11}|^2
\end{pmatrix},\label{sm:eqn:pi1}\\
\pi_2&=\begin{pmatrix}
0&0\\
0&(1-x)|u_{3,01}|^2
\end{pmatrix},\\
\pi_3&=\begin{pmatrix}
0&0\\
0&(1-x)|u_{3,11}|^2
\end{pmatrix}.
\end{align}

To conduct 3-element measurement, we set $\pi_3=0$, and replace the corresponding coupler by a light dump, which constitutes part of our current experiment device. To conduct arbitrary rank-1 3-element measurement, we insert an arbitrary gate $U_4$ in front of PPBS, and the measurement operators are transformed to $U_4^\dagger\pi_0 U_4,U_4^\dagger\pi_1 U_4,U_4^\dagger\pi_2 U_4$. To prove these operators achieve an arbitrary rank-1 3-element measurement, we first write the general form
\begin{align}
\Pi_0=\alpha_0\ket{\psi_0}\bra{\psi_0},~~\Pi_1=\alpha_1\ket{\psi_1}\bra{\psi_1},~~\Pi_2=\alpha_2\ket{\psi_2}\bra{\psi_2};~~
\mbox{with},~\alpha_i\geq 0,~\&~\sum \Pi_i=\mathbb{I}.
\end{align}

The third operator $\Pi_2$ is defined such that it can be easily related to $U_4^\dagger\pi_2 U_4$ by
\begin{align}
\alpha_2&=2-x,\\
\ket{\psi_2}&=\frac{1}{\sqrt{\alpha_2}}(\bar{u}_{4,00}+\bar{u}_{4,10}\sqrt{1-x})\ket{0}+(\bar{u}_{4,01}+\bar{u}_{4,11}\sqrt{1-x})\ket{1}.
\end{align}

Then the first two operators $\Pi_0$ and $\Pi_1$ must satisfy
\begin{align}
    \Pi_0+\Pi_1=U_4^\dagger\begin{pmatrix}
1&0\\
0&x
\end{pmatrix}U_4,
\end{align}
which is
\begin{align}
    U_4\Pi_0U_4^\dagger+U_4\Pi_1U_4^\dagger=\begin{pmatrix}
1&0\\
0&x
\end{pmatrix}.
\end{align}

Since $U_4\Pi_0U_4^\dagger$ is Hermitian and rank-1, it can be decomposed as
\begin{align}
    U_4\Pi_0U_4^\dagger=\begin{pmatrix}
        a\\b
    \end{pmatrix}\begin{pmatrix}
        \bar{a}&\bar{b}
    \end{pmatrix}.
\end{align}

Then
\begin{align}
    U_4\Pi_1U_4^\dagger=\begin{pmatrix}
1-|a|^2 & -a\bar{b}\\
-\bar{a}b & x-|b|^2
    \end{pmatrix}.
\end{align}

We find it is exactly the form of Eqn.~(\ref{sm:eqn:pi0},\ref{sm:eqn:pi1}), which leads to
\begin{align}
    \Pi_0&=U_4^\dagger\pi_0 U_4,\\
    \Pi_1&=U_4^\dagger\pi_1 U_4.
\end{align}

Once we have the capability to perform any rank-1 POVM, we can extend this to perform a general POVM by assigning mixing probabilities among a set of rank-1 POVMs.

\subsection{Choosing restaurant games}

Theoretically, there is a large set of $\mathbb{H}^3$ games that can be won quantumly but not classically, hence implying the quantum advantage. 
We set a PPBS with $\ket{V}$ reflection ratio as $1/3$ and play ten of such games, distributed on the curve
\begin{align}
\Gamma=\left\{\left(\frac{-a^2+a+4}{12-6 a},\frac{2
\left(a^2-2\right)}{3
\left(a^2-4\right)},-\frac{a^2+a-4}{
6 (a+2)}\right)\Bigg|a\in [-1,1]\right\},
\end{align}
in the parameter space of all the $\mathbb{H}^3$ games, as shown in Fig.~4 in the main text. The parameters of the chosen games include the theoretical visiting probabilities for each restaurant, the encoding states $\rho_k$ of Alice, the coefficients $\lambda_j$ of Bob's measurement operators (MO), and the rotation angles of the wave plates in $U_1,U_2,U_3$.

\begin{table*}[t]
\centering
\caption{\textbf{The parameters of chosen games.} For each played game, we provide the detailed information about the theoretical visiting probabilities for each restaurant, the encoding states $\rho_k$ of Alice, the coefficients $\lambda_j$ of measurement operators (MO) of Bob, and the rotation angles of the wave plates in $U_1,U_2,U_3$. The encoding states are represented by the vectors $\vec{n}_k$ such that $\rho_k=\ket{\psi_j}\bra{\psi_j}=(\mathbb{I}+\vec{n}_k\cdot \sigma)/2, \ k=1,2,3$. The measurement operators are $\lambda_j\ket{\psi_j^\bot}\bra{\psi_j^\bot}, \ j=1,2,3$. Over the ten games, the angles of wave plates in $U_3$ is 130.5, 355.5, 40.5.}\label{tab:data3}
\resizebox{\textwidth}{22mm}{
\begin{tabular}{|c|c|c|c|c|c|c|c|c|}
\hline
Exp. &  $\rho_1$            & $U_1$ Angles ($\rho_1$) & $\rho_2$ & $U_1$ Angles ($\rho_2$) & $\rho_3$             & $U_1$ Angles ($\rho_3$) & MO Coef.          & $U_2$ Angles      \\ \hline
1    &  0.041, 0.160, -0.986  & 298.6, 310.6, 61.4 & 0, 0, 1    & 75.0, 345.0, 75.0 & -0.117, -0.454, 0.883  & 359.1, 6.9, 0.9    & 0.986, 0.667, 0.347 & 35.6, 12.2, 48.8    \\\hline
2    & 0.248, 0.217, -0.944  & 50.3, 347.7, 309.7 & 0, 0, 1    & 75.0, 345.0, 75.0 & -0.607, -0.533, 0.589  & 102.4, 5.5, 77.6   & 0.947, 0.667, 0.386 & 341.6, 3.7, 304.1   \\\hline
3    &  0.187, -0.456, -0.870 & 126.6, 75.9, 53.4  & 0, 0, 1    & 75.0, 345.0, 75.0 & -0.370, 0.900, 0.230   & 342.7, 351.2, 17.3 & 0.885, 0.667, 0.448 & 29.8, 64.8, 33.3    \\\hline
4    & 0.100, 0.645, -0.758  & 326.0, 349.1, 34.0 & 0, 0, 1    & 75.0, 345.0, 75.0 & -0.153, -0.982, -0.107 & 22.8, 8.6, 337.2   & 0.805, 0.667, 0.528 & 304.4, 349.2, 347.4 \\\hline
5    &  0.101, -0.794, -0.600 & 121.6, 3.0, 58.4   & 0, 0, 1    & 75.0, 345.0, 75.0 & -0.116, 0.916, -0.385  & 61.9, 357.2, 118.1 & 0.714, 0.667, 0.619 & 41.1, 6.9, 2.1      \\\hline
6    &  0.906, -0.177, -0.385 & 24.2, 16.7, 335.8  & 0, 0, 1    & 75.0, 345.0, 75.0 & -0.785, 0.154, -0.600  & 331.8, 341.9, 28.2 & 0.619, 0.667, 0.714 & 357.1, 325.0, 43.4  \\\hline
7    &  0.295, -0.950, -0.107 & 343.7, 341.2, 16.3 & 0, 0, 1    & 75.0, 345.0, 75.0 & -0.194, 0.623, -0.758  & 28.6, 25.1, 331.4  & 0.528, 0.667, 0.805 & 1.1, 43.9, 332.6    \\\hline
8    &  -0.001, 0.973, 0.230  & 351.1, 342.7, 8.9  & 0, 0, 1    & 75.0, 345.0, 75.0 & 0.001, -0.493, -0.870  & 28.5, 31.0, 331.5  & 0.448, 0.667, 0.885 & 255.8, 304.3, 326.9 \\\hline
9    &  0.731, 0.344, 0.589   & 77.3, 4.7, 102.7   & 0, 0, 1    & 75.0, 345.0, 75.0 & -0.298, -0.140, -0.944 & 129.9, 79.7, 50.1  & 0.386, 0.667, 0.947 & 127.9, 89.6, 25.1   \\\hline
10    & -0.465, 0.0610, 0.883 & 356.4, 6.0, 3.6    & 0, 0, 1    & 75.0, 345.0, 75.0 & 0.163, -0.0214, -0.986 & 40.4, 330.6, 319.6 & 0.347, 0.667, 0.986 & 52.0, 337.9, 310.3  \\ \hline
\end{tabular}}
\end{table*}

\subsection{Minimization of classical value of the quality index $\mathcal{E}$}

In this section we describe a general method to compute the minimum values of $\mathcal{E}_C$ attainable by classical strategies. Consider a $n$-restaurant game where Alice is allowed to communicate $\log_2d$-bits to Bob. The most general classical strategy that can be employed by Alice and Bob is described as follows.
\begin{enumerate}
    \item When the $i^{\text{th}}$ restaurant is closed, Alice sends Bob the message $j \in \{1,2,\cdots, d\} $ with probability $p^A(j|i_c)$.
    \item Upon receiving message $j$, Bob visits the $k^{\text{th}}$ restaurant with probability $p^B(k_b|j)$.
    
\end{enumerate}
The probability $p(k_b|i_c)$ of Bob visiting the $k^{\text{th}}$ restaurant when the $i^{\text{th}}$ restaurant is closed is then given by $p(k_b|i_c) = \sum_j p^A(j|i_c) p^B(k_b|j)$.\\

We recall that $\mathcal{E}$ is defined as,
\begin{equation}
    \mathcal{E}:= \max_j \left\{k_1 \sum_i p(i_b|i_c), ~k_2|\gamma_j-p_j| \right\}, \label{def:eps}
\end{equation}
where $k_1$ and $k_2$ are given constants that depend on the relative importance of conditions (h1) and (h2) respectively, and $p_j = \sum_i p(j_b|i_c) p(i_c) = \frac{1}{n}\sum_i p(j_b|i_c)$. For a given probability distribution $p^A(j|i_c)$, the minimum value of $\mathcal{E}$ can be obtained by solving the following optimization problem:
\begin{align}
    \min_{ \mathcal{E},\{p^B(k_b|j)\}} {\mathcal{E}} ~~~~~~~~ &~~ \text{subject to} \nonumber\\
    & p^B(k_b|j) \ge 0 ~\forall k,j ,\nonumber\\
    &\sum_k p^B(k_b|j) = 1  ~\forall j,\nonumber\\
    & k_1 \sum_i p(i_b|i_c) \le \mathcal{E} ,\nonumber\\
    & k_2 (\gamma_j-p_j) \le \mathcal{E} ~\forall j ,\nonumber\\
    & k_2 (\gamma_j-p_j) \ge -\mathcal{E} ~\forall j  ,\label{eq:lp}
\end{align}
where, $p(k_b|i_c) = \sum_j p^A(j|i_c) p^B(k_b|j)$ and $p_j =\frac{1}{n}\sum_i p(j_b|i_c)$.
The first two conditions require that $p^B(k_b|j)$ is a valid probability distribution, and the last three conditions specify that $\mathcal{E}$ is defined by \eqref{def:eps}.
The optimal classical value $\mathcal{E}_C$ is obtained by solving the above optimization problem for all the allowed probability distributions $ p^A(j|i_c) $ and taking the minimum value. The optimization problem presented in \eqref{eq:lp} is a linear program, for which there are efficient numerical algorithms that are guaranteed to converge to the optimal value within a precision of $10^{-3}$. For the $3$-restaurant game, We performed this optimization by discretizing the allowed values of $p^A(j|i_c)$. We selected the parameters $k_1 = \frac{1}{3}$ and $k_2 = 1$.

\subsection{The classical and experimental quantum winning index values}

We list the winning index values of both classical and implemented quantum strategies in Tab.~\ref{tab:eps}.

\begin{table*}[h!]
\centering
\caption{The classical and experimental quantum winning index values, as well as the ideal and experimental visiting probabilities for each restaurant.}\label{tab:eps}
\resizebox{0.85\textwidth}{24mm}{
\begin{tabular}{|c|c|c|c|c|c|}
\hline
Exp. & quantum (ideal)     & quantum (noisy)         & classical  & $\gamma_i$ (ideal)  &$p_i$ (noisy)  \\ \hline
1    & 0 & 0.00348  & 0.00658 & 0.64694, 0.23368, 0.11938   & 0.6483 , 0.2358 , 0.1159       \\\hline
2    & 0 & 0.01592  & 0.01770  & 0.59454, 0.26168, 0.14378   & 0.57862, 0.26316, 0.15822      \\\hline
3    & 0 & 0.0135  & 0.03517  & 0.52361, 0.29331, 0.18308   & 0.5253 , 0.27981, 0.19489    \\\hline
4    & 0 & 0.00455  & 0.05375   & 0.44626, 0.31831, 0.23543  & 0.44532, 0.3147 , 0.23998    \\\hline
5    & 0 & 0.00732  & 0.07264  & 0.36982, 0.33164, 0.29854   & 0.37686, 0.32432, 0.29882   \\\hline
6    & 0 & 0.00947  & 0.07264    & 0.29854, 0.33164, 0.36982  & 0.29803, 0.32268, 0.37929   \\\hline
7    & 0 & 0.00508  &  0.05375   & 0.23543, 0.31831, 0.44626   & 0.23035, 0.31992, 0.44973  \\\hline
8    & 0 & 0.00978  &0.03517  & 0.18308, 0.29331, 0.52361   & 0.18581, 0.28353, 0.53066  \\\hline
9    & 0 & 0.00755  & 0.01770  & 0.14378, 0.26168, 0.59454   & 0.13986, 0.25804, 0.60209 \\\hline
10    & 0 & 0.00046 & 0.00658   & 0.11938, 0.23368, 0.64694  & 0.11985, 0.23343, 0.64672   \\\hline
\end{tabular}}
\end{table*}

\section{Device Certification}

Certifying the quantumness of devices is essential for advancing quantum technologies. This certification process is a prerequisite for many quantum applications. Utilizing the non-classical phenomena observed in restaurant games, we can effectively validate the quantumness of both state preparation devices (PD) and measurement devices (MD). In this work, we present a comprehensive procedure for certifying the quantumness of these devices through the restaurant game framework. The workflow is depicted in Fig.~\ref{fig4}.
\begin{itemize}
\item[] {\bf Step 1:} The state preparation device (PD) at Alice's end generates an ensemble \(\mathcal{S} := \{\rho_i \mid i = 1, \ldots, n\}\) of qubit states, where each state is encoded as \(\rho_i = \ket{\psi_i}\bra{\psi_i}\). Here, \(\ket{\psi_i}\) represents the encoding states used in an \(n\)-restaurant game known to exhibit quantum advantage.
\item[] {\bf Step 2:} The ensemble is then transferred to Bob, who implements a measurement device (MD) characterized by a set of positive-operator-valued measures (POVM) \(\{M_k\}_k\). Each POVM element \(M_k\) is given by \(M_k = \alpha_k \ket{\psi_i^\bot}\bra{\psi_i^\bot}\), where \(\alpha_k\) is determined by the parameters of the restaurant game.
\item[] {\bf Step 3:} The verifier collects the data from the state preparation and measurement processes and computes the winning index function \(\mathcal{E}_\text{V}\) for the restaurant game. Simultaneously, the verifier calculates the minimum winning index function \(\mathcal{E}_\text{C}\) for classical strategies by solving the optimization problem defined in Eq.~(\ref{eqn:E_best}). If \(\mathcal{E}_\text{V} < \mathcal{E}_\text{C}\), the verifier confirms the quantumness of the PD and MD pair; otherwise, the certification process fails.
\end{itemize}
The following proposition demonstrates the effectiveness of our device certification protocol.
\begin{figure}[t!]
\begin{center}
\includegraphics[width=.5\linewidth]{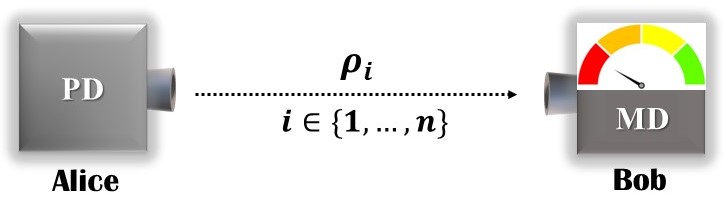}
\end{center}
\caption{\textbf{Nonclassicality of (PD,MD).} A PD at Alice's end produce an ensemble $\mathcal{S}:=\{\rho_i~|~i=1,\cdots,n\}$ of qubit states. The system is transferred to Bob who implements an MD. Quantum advantage in $\mathbb{H}^3(\gamma_1,\gamma_2,\gamma_3)$ game certify successful implementation of such (PD, MD) pairs.}\label{fig4}
\end{figure}

\begin{prop}
Quantum advantage in \(\mathbb{H}^3(\gamma_1, \gamma_2, \gamma_3)\) games guarantees both the presence of coherence in the encoded qubit states and the implementation of a non-projective measurement by the receiver.
\end{prop}
\begin{proof}
We first establish that a quantum strategy, where Alice uses orthogonal encoding states, can be simulated by a classical strategy. In this quantum scenario, Alice transmits the encoded state \(\rho_i = \sum_j p^A(j \mid i_c) \ket{j}\bra{j}\) to Bob when the \(i\)-th restaurant is closed, where \(\{\ket{j}\}_{j=1}^2\) forms an orthonormal basis in \(\mathbb{C}^2\). Bob then performs a measurement with POVM elements \(\{M_k\}_k\) for decoding and visits the \(k\)-th restaurant based on the measurement outcome. The probability that Bob visits the \(k\)-th restaurant when the \(i\)-th restaurant is closed is given by: $p(k_b \mid i_c) = \sum_j p^A(j \mid i_c) \bra{j} M_k \ket{j}$. This quantum strategy can be effectively simulated by a classical strategy. In this classical simulation, Alice sends the message \(j \in \{1,2\}\) with probability \(p^A(j \mid i_c)\) when the \(i\)-th restaurant is closed. Upon receiving the message \(j\), Bob visits the \(k\)-th restaurant with probability \(\bra{j} M_k \ket{j}\).

Next, we demonstrate that a projective measurement performed during the decoding step can be simulated by a classical strategy. Consider the scenario where Alice sends the state \(\rho_i\) to Bob when the \(i\)-th restaurant is closed. Bob then performs a projective measurement \(\{\ket{j}\bra{j}\}_{j=1}^2\) and visits the \(k\)-th restaurant with probability \(p^B(k_b \mid j)\) upon obtaining outcome \(j\). The probability that Bob visits the \(k\)-th restaurant when the \(i\)-th restaurant is closed is given by: $p(k_b \mid i_c) = \sum_j p^B(k_b \mid j) \bra{j} \rho_i \ket{j}$. This quantum strategy can be simulated by a classical strategy as follows: Alice sends the message \(j \in \{1,2\}\) with probability \(\bra{j} \rho_i \ket{j}\) when the \(i\)-th restaurant is closed. Upon receiving the message \(j\), Bob then visits the \(k\)-th restaurant with probability \(p^B(k_b \mid j)\). This completes the proof.
\end{proof}


\bla

